\documentclass[submission,copyright,creativecommons]{eptcs}
\providecommand{\event}{CL\&C'18} 
\usepackage{breakurl}             
\usepackage{underscore}           
\usepackage[english]{babel}
 \usepackage[utf8x]{inputenc}
 \usepackage{amssymb}

\usepackage[all]{xy}
\usepackage[normalem]{ulem}
\usepackage{pgfplots}

  \usepackage{alltt}

 \usepackage{enumitem}   

\usepackage[cmex10]{amsmath}
 \usepackage{amsthm}

 \usepackage{graphicx}  

 \usepackage{hyperref}

\newcommand*{\nolink}[1]{%
#1
}

\def\mybs{\char092}
 \newtheorem{lemma}{Lemma}
 \newtheorem{notation}{Notation}
 \newtheorem{theorem}{Theorem}
 \newtheorem{example}{Example}
 \newtheorem{definition}{Definition}
 \newtheorem{remark}{Remark}

\def\titlerunning{Validating Back-links of FOL$_{\mbox{\normalsize ID}}$ Cyclic Pre-proofs}
\def\authorrunning{S. Stratulat}

\title{Validating Back-links of FOL$_{\mbox{\normalsize ID}}$ Cyclic Pre-proofs\\
}
\author{Sorin Stratulat
\institute{
Université de Lorraine, CNRS, LORIA\\
F-57000 Metz, France}
\email{sorin.stratulat@univ-lorraine.fr}
}

\begin{document}

\maketitle

\begin{abstract} 

  Cyclic pre-proofs can be represented as sets of finite tree
  derivations with back-links. In the frame of the first-order logic
  with inductive definitions (FOL$_{\mbox{\scriptsize ID}}$), the
  nodes of the tree derivations are labelled by sequents and the
  back-links connect particular terminal nodes, referred to as
  \emph{buds}, to other nodes labelled by a same sequent. However,
  only some back-links can constitute sound pre-proofs. Previously, it
  has been shown that special ordering and derivability conditions,
  defined along the minimal cycles of the digraph representing a
  particular normal form of the cyclic pre-proof, are sufficient for
  validating the back-links. In that approach, a same constraint could
  be checked several times when processing different minimal cycles,
  hence one may require additional recording mechanisms to avoid
  redundant computation in order to downgrade the time complexity to
  polynomial.

  We present a new approach that does not need to process minimal
  cycles. It based on a normal form that allows to define the
  validation conditions by taking into account only the root-bud paths
  from the non-singleton strongly connected components of its digraph.

\end{abstract}

\section{Introduction}
 
In~\cite{Brotherston:2005qy,Brotherston:2006uy,Brotherston:2011fk},
Brotherston and Simpson introduced the notion of cyclic (pre-)proof in
the frame of first-order logic with inductive definitions (for short
FOL$_{\mbox{\scriptsize ID}}$ and detailed, e.g.,
in~\cite{Aczel:1977fk}) and including equality. In this setting, the
cyclic \emph{pre-proofs} are sequent-based proof derivations usually
presented in the form of finite trees. Some of their terminal nodes,
called \emph{buds}, are labelled by `not-yet proved' sequents that
already labelled other nodes, called \emph{companions}. For each bud
there is only one companion and the bud-companion relations are
referred to as \emph{back-links}.

Not all back-links may constitute sound pre-proofs. Indeed, a
pre-proof can be constructed for any false sequent $S$ by applying a
stuttering inference step\footnote{for example, by applying the LK's
  $(Subst)$ rule with an identity substitution (see the definition of
  $(Subst)$ in Definition~\ref{def:trace}).} that creates a copy of
$S$. This terminal node is a bud whose companion is the root of the
pre-proof. \cite{Brotherston:2006uy,Brotherston:2011fk} also
introduced the CLKID$^{\omega}$ inference system for building cyclic
pre-proofs and defined a sufficient criterion for checking their
soundness in terms of a \emph{global trace condition}. This condition
is an $\omega$-regular property that can be checked as an inclusion
between two Büchi automata. The inclusion test includes an automata
complementation procedure~\cite{Kupferman:2001aa} whose time
complexity is exponential in the number of states of the automaton to
be complemented.

A more effective soundness criterion was
given in~\cite{Stratulat:2017ac} for pre-proofs generated by
CLKID$^{\omega}_N$, a restricted version of CLKID$^{\omega}$.  Inspired from a previous
method~\nolink{\cite{Stratulat:2012uq,Stratulat:2017aa}} for checking
the soundness of cyclic proofs built using the Noetherian induction
principle for reasoning on conditional specifications, its time
complexity can be downgraded to polynomial. To do this, a
CLKID$^{\omega}_N$ pre-proof is normalised to some set of finite tree
derivations which can be represented as a directed graph (for short, digraph)
having some arrows labelled by substitutions. The soundness criterion
asks that some derivability and ordering
constraints hold along the paths
leading root nodes to bud nodes in the \emph{minimal} cycles of the digraph, i.e.,
cycles that do not include other cycles.

In general, the number of minimal cycles in a digraph with $n$ nodes
can be much greater than the number of its buds (which is always
smaller than $n$). For complete digraphs, i.e., digraphs for which
every pair of distinct nodes is connected by arrows in the two ways,
one can define the number of minimal cycles built by $k\in[2..n]$
nodes, as follows. We take one of the $n$ nodes as the starting node
in the cycle, then the next one from the remaining $n-1$ nodes, and so
on for $k-1$ times. So there are
$n\times (n-1)\times\ldots\times (n-k+1)$ ways to do it. Since the
cycle consisting of the $k$ nodes can be built in $k$ different times,
depending which is the starting node among its nodes, this number
is $\frac{n!}{(n-k)!k}$.  Hence, the total number of minimal
cycles in a complete digraph with $n$ nodes is

\[\sum_{k=2}^n \frac{n!}{(n-k)!k}\]

Fortunately, the number of arrows in any digraph built with the
approach from~\cite{Stratulat:2017ac} is smaller than that for the
complete digraphs because each bud node has only one
companion. However, a ordering-derivability constraint can be checked
several times as it may be defined w.r.t. different minimal
cycles. In~\cite{Stratulat:2017ac}, it was already suggested that
their number can be reduced to the number of buds from the minimal
cycles, hence smaller than $n$. This redundancy can be avoided, for
example, by using recording mechanisms.

In this paper, we present an improved version of the soundness criterion for
validating CLKID$^{\omega}_N$ pre-proofs. The advantage is that the
computation of minimal cycles is not needed and there is no redundancy
in the computation of the ordering-derivability constraints. In order
to do this, we propose a new normal form of CLKID$^{\omega}_N$
pre-proofs and define ordering and derivability constraints for every
root-bud path that occurs in a non-singleton strongly connected
component (SCC) of the digraphs associated to the new normal forms. We
show that the number of constraints is
that of the buds from the non-singleton
SCCs. 

The rest of the paper is structured as
follows. Section~\ref{sec:sequent} gives a brief presentation of
FOL$_{\mbox{\scriptsize ID}}$ and
CLKID$^{\omega}_N$. Section~\ref{sec:proofs} introduces the soundness
criterion for CLKID$^{\omega}_N$ pre-proofs by detailing the
normalisation procedure, the digraph construction and the
definition of the ordering and derivability conditions. A comparison
is made with the soundness criterion
from~\cite{Stratulat:2017ac}. The conclusions and future work are
given in the last section.

\section{Induction-based sequent calculus}
\label{sec:sequent}

\noindent\textbf{Syntax.} The logical setting is that presented
in~\cite{Brotherston:2011fk},
based on FOL$_{\mbox{\scriptsize ID}}$ with equality using a standard
(countable) first-order language $\Sigma$. The predicate symbols are
labelled either as \emph{ordinary} or \emph{inductive}, and we assume
that there is an arbitrary but finite number of inductive predicate
symbols. 
The terms are defined as usual. By $\overline{t}$,
we denote a vector of terms ($t_1, \ldots, t_n$) of length $n$, the
value of $n$ being usually deduced from the context.

New terms and formulas are built by instantiating variables by terms
via substitutions. A \emph{substitution} is a mapping from variables
to terms, of the form $\{x_1\mapsto t_1; \ldots ;x_p \mapsto t_p\}$,
for some $p>0$, which can be
written in a more compact form as $\{\overline{x} \mapsto
\overline{t}\}$, where $\overline{x}\equiv (x_1,\ldots,x_p)$, 
$\overline{t}\equiv (t_1,\ldots,t_p)$, and $\equiv$ is the syntactic
equality.
The \emph{composition} of $\sigma_1$ with
$\sigma_2$ is denoted by $\sigma_1\sigma_2$, for all substitutions
$\sigma_1$ and $\sigma_2$. 
A term $t$ is an \emph{instance} of $t'$, or $t$ \emph{matches} $t'$, if
there is a substitution $\sigma$ such that $t\equiv t'\sigma$.   
Similarly, the notion of matching 
can be
extended to vector of terms,
atoms, and formulas. For any substitution
$\sigma$ 
applied to a formula $F$, we use the notation $F[\sigma]$ instead of
$F\sigma$.

\centerline{}\noindent\textbf{Deductive sequent-based inference rules.} The proof
derivations are built from sequents~\cite{Gentzen:1935fk} of the form
$\Gamma \vdash \Delta$, where $\Gamma$ and $\Delta$ are finite
multisets of formulas called \emph{antecedents} and \emph{succedents},
respectively. $FV(\Gamma \vdash \Delta)$ denotes  its set of
free variables. An inference rule is represented by a horizontal line
followed by the name of the rule. The line separates the lower sequent, called
\emph{conclusion}, from a (potentially empty) multiset of upper sequents,
called \emph{premises}. Most of the rules \emph{introduce} an
explicitly represented
formula from the conclusion, called \emph{principal} formula. In this
case, the rules are annotated by $L$ (resp., $R$) if the rule is
introduced on the left (resp, right) of the $\vdash$ symbol from the conclusion.

A specification is built from a finite inductive definition set of
axioms $\Phi$ consisting of formulas of the form
\begin{eqnarray}
\label{eq:implication}\bigwedge_{m=1}^hQ_{m}(\overline{u}_{m})\wedge \bigwedge_{m=1}^l P_{m}(\overline{t}_{m}) \Rightarrow P(\overline{t}),&
\end{eqnarray}
\noindent where $h, l$ are
naturals, $Q_{1},\ldots, Q_{h}$ are ordinary predicate symbols,
$P_{1},\ldots,P_{l}, P$ are inductive predicate symbols.
$\bigwedge_{m=1}^0$ is a shortcut for the `true' boolean
constant and can be ignored.

The deductive
 part of the sequent-based reasoning about
 FOL$_{\mbox{\scriptsize ID}}$ is performed using the Gentzen's LK
 rules~\cite{Gentzen:1935fk} and an `unfold' 
 rule. The unfold rule $(R.(rname))$ replaces an atom $P(\overline{t}')$
using the axiom $(rname)$ defining $P$. E.g., 
(\ref{eq:implication}) can be applied on $\Gamma\vdash
P(\overline{t}'),\Delta$ if $P(\overline{t}')\equiv
P(\overline{t})[\sigma]$ for some substitution $\sigma$, as:
\begin{center}
    \includegraphics[width=0.9\linewidth]{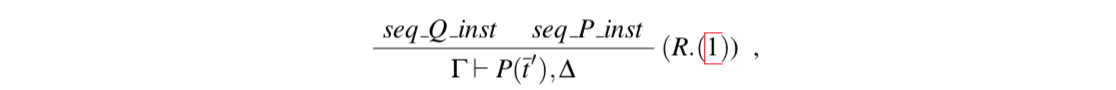}
  \end{center}

\noindent where $seq\_Q\_inst$ (resp., $seq\_P\_inst$) is the multiset of sequents
$\bigcup_{m=1}^h\{\Gamma\vdash Q_{m}(\overline{u}_{m})[\sigma],\Delta\}$
\sloppy{(resp., $\bigcup_{m=1}^l\{\Gamma\vdash
  P_{m}(\overline{t}_{m})[\sigma],\Delta\}$)}. 

\centerline{}\noindent \textbf{Semantics.} The standard interpretation of inductive
predicates is built from prefixed points of a monotone operator issued
from the set of axioms representing
$\Phi$~\cite{Aczel:1977fk}. Its least prefixed point, approached by an
iteratively built \emph{approximant} sequence, helps defining a
\emph{standard model} for $(\Sigma, \Phi)$ (see, e.g.,~\cite{Brotherston:2011fk} for
details).

\begin{definition}[validity]
\label{def:seq-validity}
Let $M$ be a standard model for $(\Sigma,\Phi)$, $\Gamma\vdash\Delta$
a sequent and $\rho$ a valuation which interprets in $M$ the
variables from $FV(\Gamma \vdash \Delta)$. We write $\Gamma
\models_{\rho}^M\Delta$ if every $G\in\Gamma$
holds in $M$ there is some
$D\in\Delta$ that also holds in $M$. We say that $\Gamma\vdash\Delta$ is
$M$-\emph{true} and write $\Gamma \models^M\Delta$ if
$\Gamma\models_{\rho}^M\Delta$, for any $\rho$.
\end{definition}

\noindent When $M$ is
  implicit from the context,  we use \emph{true} instead of
  $M$-true. 
A rule is \emph{sound}, or preserves the
validity, if its conclusion is true whenever its premises are
true. Hence, the conclusion of every $0$-premise sound rule is true. 
 
\subsection{The CLKID$_N^ {\omega}$ cyclic inference system}

CLKID$^{\omega}$~\cite{Brotherston:2011fk} includes the LK rules, the
rules from Figure~\ref{fig:eq} that process equalities, the `unfold'
rule and the $(Case)$ rule which represents a left-introduction
operation for inductive predicate symbols:

\begin{center}
    \includegraphics[width=0.9\linewidth]{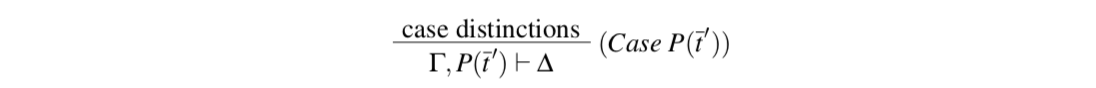}
  \end{center}

\noindent
For each axiom of the
form~(\ref{eq:implication}),
\begin{eqnarray}
\label{eq:case}
\Gamma, \overline{t}'=\overline{t},Q_{1} (\overline{u}_{1}),\ldots,
  Q_{h}(\overline{u}_{h}),P_{1}(\overline{t}_{1}),\ldots,
  P_{l}(\overline{t}_{l})\vdash\Delta&
\end{eqnarray}
\noindent is the \emph{case distinction}  for which each free variable $y$ from (\ref{eq:implication}) is fresh
w.r.t. the \emph{free variables} from the conclusion of the rule ($y$ can be renamed to a fresh variable, otherwise). $P_{1}(\overline{t}_{1}),\ldots,
  P_{l}(\overline{t}_{l})$ are \emph{case
  descendants} of $P(\overline{t}')$. 

\begin{figure*}[!ht]

  \begin{center}
    \includegraphics[width=0.9\linewidth]{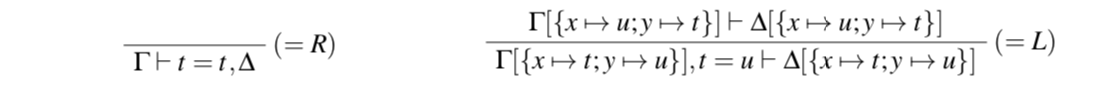}
  \end{center}

\caption{\label{fig:eq} Sequent-based rules for equality reasoning.}
\end{figure*}

The inference system \label{def:clkida} CLKID$_N^{\omega}$, introduced
in~\cite{Stratulat:2017ac}, is the restricted version of
CLKID$^ {\omega}$ for which $(=L)$ is
replaced by the
generalization rule $(Gen)$ that substitutes a term by a variable:

\begin{center}
    \includegraphics[width=0.9\linewidth]{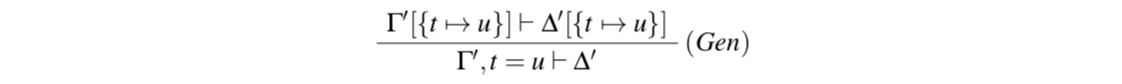}
  \end{center}

\noindent where $t$ is a free variable that does not occur in $u$. 

$(Gen)$ is the particular instance of $(=L)$ from
Figure~\ref{fig:eq} when $y\not\in FV(\Gamma\vdash\Delta)$ and
$t\not\in FV(\Gamma\vdash\Delta)$ that also does not occur in $u$. By
using the property that $\Phi[\{y\mapsto u\}]\equiv \Phi$, holding
whenever $y$ is a free variable not occurring in a formula $\Phi$, the
last condition can simplify $(=L)$ to a form equivalent to $(Gen)$:

\begin{center}
    \includegraphics[width=0.9\linewidth]{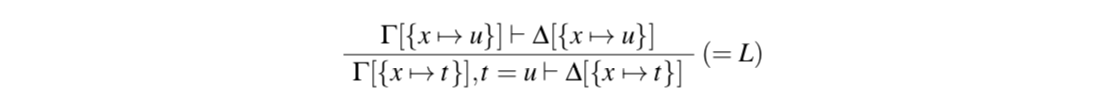}
  \end{center}

\centerline{}\noindent\textbf{CLKID$_N^{\omega}$ pre-proof trees.} A
\emph{derivation tree} for a sequent $S$ is built by successively applying
inference rules starting from $S$. We
consider only finite derivation trees whose terminal nodes can be
either leaves or buds. The \emph{leaves} are labelled by sequents that
represent conclusions of 0-premise rule, e.g., the unfold
rule using unconditional axioms. For each \emph{bud} there is
another node, called \emph{companion} and having the same
sequent labelling. The bud and its companion are annotated by the same sign, e.g.,
$\boldsymbol{\dag}$. In addition, the buds having a same companion
are labelled by the sign followed by a number that makes them unique, e.g.,
$\boldsymbol{\dag 1}$, $\boldsymbol{\dag 2}$,\ldots. A
\emph{back-link} is a relation bud-companion.

\begin{notation}[pre-proof tree, induction function for tree]
\label{def:tpreproof}
The pair ($\cal D$, $\cal R$) denotes a \emph{pre-proof tree}, where
$\cal D$ is a finite derivation tree and $\cal R$ is a defined
\emph{induction function} assigning a companion to every bud in
$\cal D$.
\end{notation}


\begin{example}\label{ex:preproof} To highlight the changes w.r.t. \cite{Stratulat:2017ac}, we
  take the same running example (also presented in~\cite{Brotherston:2012fk}). Let $N$ and $R$ be two inductive predicates defined by:

  \begin{tabular}{cc}
    \begin{minipage}[b]{.4\linewidth}
      {\begin{align}
        \label{n0} \Rightarrow N(0) & \\
        \label{n1} N(x)\Rightarrow N(sx)&
      \end{align}}
    \end{minipage}
                                          &
\begin{minipage}[b]{.4\linewidth}
  {\begin{align}
    \label{r0} \Rightarrow R(0,y) &\\
    \label{r1} R(x,0)\Rightarrow R(sx,0)&\\
    \label{r2} R(ssx,y) \Rightarrow R(sx, sy)&
  \end{align}}
\end{minipage}
  \end{tabular}
  
\noindent where the parentheses around the argument of $s$ are omitted. One can build the following pre-proof of $N(x),N(y)\vdash R(x,y)$:

\begin{center}
    \includegraphics[width=0.9\linewidth]{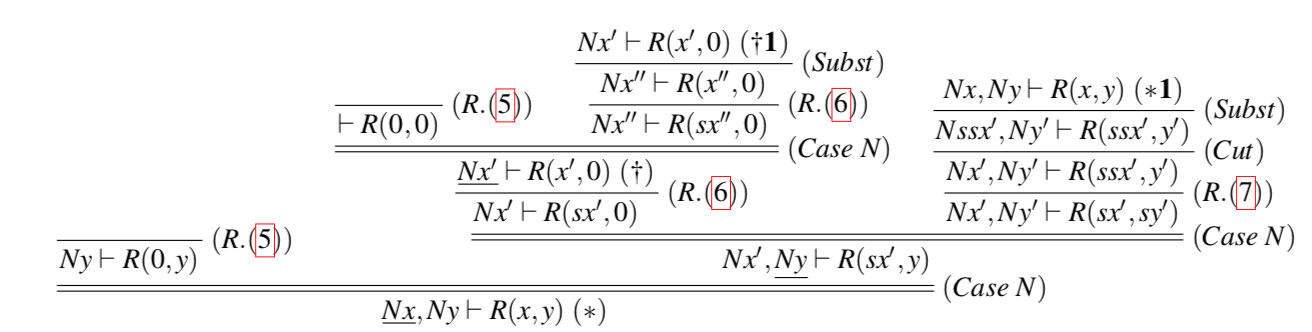}
  \end{center}

where the double line means that $(Gen)$ was applied after
$(Case)$ and the principal formulas of the $(Case)$ steps are
underlined. $(Cut)$ is applied as in~\cite{Brotherston:2012fk}. For
lack of space, the parentheses around the argument of $N$ are omitted.
\end{example}

We denote by $S(N)$ the sequent labelling any node $N$. A \emph{path}
is a list $[N^0, N^1, \ldots]$ of nodes in a pre-proof tree such that,
for all $i\geq 0$, $S(N^{i+1})$ is either one of the premises of the
rule applied on $S(N^i)$ if $N^i$ is not a terminal node, or
$S({\cal R}(N^i))$ if $N^i$ is a bud.

\begin{definition}[Trace, Progress point~\cite{Stratulat:2017ac}]
\label{def:trace}
Let ($\cal D$, $\cal R$) be a CLKID$^\omega_N$ pre-proof tree and let
$[N^0, N^1, \ldots]$ be one of its infinite paths and denoted by $l$. A \emph{trace}
following $l$ is a sequence $(\tau_i)_{i\geq 0}$ of inductive
antecedent atoms (IAAs) such
that, for all $i$, we have that $N^i$ is labelled by
$\Gamma_i\vdash \Delta_i$ and:
\begin{enumerate}
\item $\tau_i$ is some $P_{j_i}(\overline{t_i})\in\Gamma_i$;
\item if $\Gamma_i\vdash\Delta_i$ is the conclusion of $(Subst)$ then
  $\tau_i=\tau_{i+1}[\theta]$, where $\theta$ is the substitution
  used by 
the LK's $(Subst)$ rule defined as:
\begin{center}
    \includegraphics[width=0.9\linewidth]{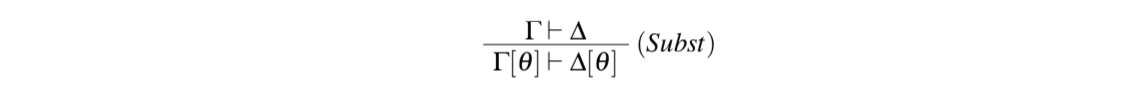}
  \end{center}

\item if $\Gamma_i\vdash\Delta_i$ is the conclusion of $(Gen)$ having
  $t=u$ as principal formula, there
  is a formula $F$ such that $\tau_i=F$ and $\tau_{i+1}=F[\{t\mapsto u\}]$;
\item if $\Gamma_i\vdash\Delta_i$ is the conclusion of a $(Case)$
  rule then either a) $\tau_{i+1}=\tau_i$, if $\tau_i$ is not the
  principal formula of the rule instance, or b) $\tau_i$ is the
  principal formula  and $\tau_{i+1}$ is a case
  descendant of $\tau_i$. In the latter case, $i$ is said to be a
  \emph{progress point} of the trace;
\item if $\Gamma_i\vdash\Delta_i$ is the conclusion of any other rule
  then $\tau_{i+1}=\tau_i$.
\end{enumerate}
\end{definition}

\begin{remark}
\label{rem:trace} Non-equality relations between (instances of) $\tau_i$ and
$\tau_{i+1}$ in the above definition are possible only when $i$ is a
progress point.
\end{remark}
\begin{remark}
  \label{rem:equivtrace} Condition 3 is an abbreviated form of the case
  dealing with $(=L)$ in Definition 5.4
  from~\cite{Brotherston:2011fk}, by applying the discussed restrictions to
  $(=L)$, i.e., if $\Gamma_i\vdash\Delta_i$ is the conclusion of $(=L)$, of
  the form
  $\Gamma [\{x\mapsto t; y\mapsto u\}],t=u\vdash \Delta [\{x\mapsto t;
  y\mapsto u\}]$ and having $t=u$ as principal formula, there is a
  formula $F'$ such that $\tau_i=F'[\{x\mapsto t; y\mapsto u\}]$ and
  $\tau_{i+1}=F'[\{x\mapsto u; y\mapsto t\}]$ under the following
  conditions:
  $y \not \in FV(\Gamma_i \backslash \{t=u\}\vdash\Delta_i)$, $t$ is a
  free variable not occurring in $u$ and
  $t\not\in FV(\Gamma\vdash\Delta)$.
\end{remark}

An \emph{infinitely progressing} trace is a trace with infinitely
many progress points.

\section{The criterion for validating the soundness of CLKID$_N^\omega$ pre-proofs} 
\label{sec:proofs}

The proof that some cyclic pre-proof is sound  is done by using a \emph{Descente Infinie}
argument. The general technique is to assume, by contradiction, that the
root of a pre-proof is labelled by
a false sequent. Then, we have to show that there is an infinite path
of nodes in the pre-proof for which there
is an infinite progressing trace following some tail of it. This means
that all successive steps in the tail are decreasing and the steps
corresponding to the progress points are strictly decreasing
w.r.t. some semantic ordering over ordinals. We get a contradiction
because it is not possible to built an infinite strictly decreasing
sequence of ordinals.

Since the inference rules are sound, an infinite path of nodes
labelled by false sequents should exist in the pre-proof whenever its
root sequent is false. A sufficient criterion for validating the
soundness of CLKID$^\omega$ pre-proofs is the \emph{global trace
  condition}~\cite{Brotherston:2005qy,Brotherston:2006uy,Brotherston:2011fk}:
for every infinite path, there is an infinitely progressing trace
following some tail. A different sufficient criterion for validating
the soundness of CLKID$_N^\omega$ pre-proofs was given
in~\cite{Stratulat:2017ac}; it defines ordering and derivability
conditions to be satisfied by the digraph representing some normal
form of the pre-proof. The normalisation procedure transforms the
pre-proof into a set of pre-proof trees, for short \emph{pre-proof
  tree-sets}, such that the root of the pre-proof is among the roots
of the trees from the normal form. If the sequent labelling the root
of the pre-proof is false, one can build an infinite path in the
digraph, whose nodes are labelled by false sequents and for which
there is an infinite progressing trace following some tail of it.

In the following, we present an
improved version of the criterion from~\cite{Stratulat:2017ac}.

\subsection{Normalising pre-proof trees}

The
normalisation process consists in the exhaustive application of the
following three operations. 
The first operation applies on an internal node labelled by some premise
of $(Subst)$, of the form
\begin{center}
    \includegraphics[width=0.9\linewidth]{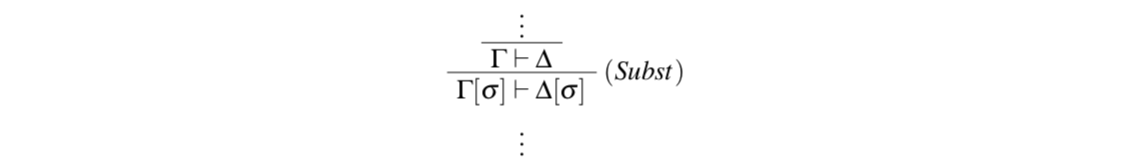}
  \end{center}

\noindent 
The result is displayed in Figure~\ref{fig:first-tr}.  The internal node is
duplicated and the subtree derivation rooted by it is detached
to become a new tree derivation. At the end, we get two distinct
pre-proof trees.  The two occurrences of the duplicated node establish
a new bud-companion relation.

\begin{figure}[!ht]
  \begin{center}
    \includegraphics[width=0.9\linewidth]{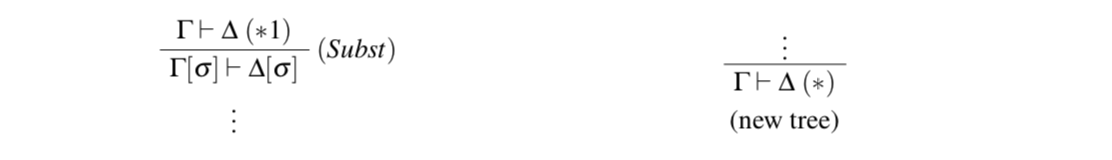}
  \end{center}
\caption{\label{fig:first-tr} The result of the first operation.}
\end{figure}

The second operation applies on a  non-root
companion  which is duplicated and the subtree derivation rooted
by it becomes
a new pre-proof tree.
The result is displayed in Figure~\ref{fig:second-tr}. The sequent
labelling the copy of the companion $(*)$ becomes the
conclusion of a new $(Subst)$ rule. The substitution used by the new
$(Subst)$ rule is chosen such that its
premise labels a new bud node labelled by the same sequent as the
conclusion, e.g., the \emph{empty} substitution. The new bud node will
have (*) assigned as companion.

\begin{figure}[!ht]
%
  \begin{center}
    \includegraphics[width=0.9\linewidth]{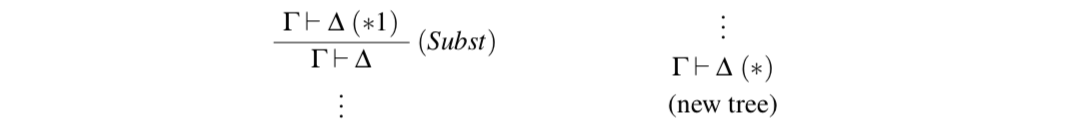}
  \end{center}

\caption{\label{fig:second-tr} The result of the second operation.}
\end{figure}

The last operation applies on a bud node labelled by some sequent that is the
premise of a rule $r$ different from $(Subst)$ such that

\begin{center}
    \includegraphics[width=0.9\linewidth]{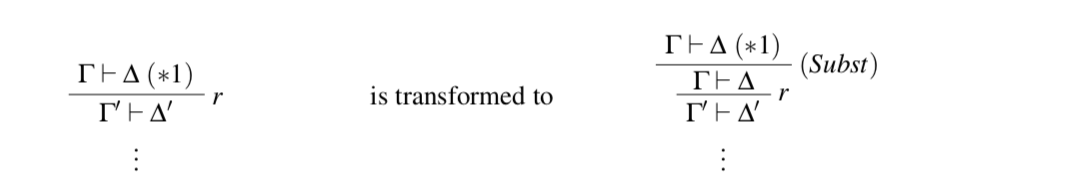}
  \end{center}

Let $(*)$ denote the companion of the bud node.  A new application of $(Subst)$ with the empty substitution was
performed on the bud sequent such that the node labelled by its
premise becomes the new bud node whose  companion is $(*)$.

Compared with the normalisation procedure
from~\cite{Stratulat:2017ac}, the two procedures  share only the
first operation. The procedure from~\cite{Stratulat:2017ac} also includes an operation that applies on
non-root companions but does not include the $(Subst)$-step from
Figure~\ref{fig:second-tr}. It does not have an equivalent transformation for the third
operation.

The following properties, related to the normalisation process and the
resulted normal form as given by
Lemmas~\ref{lem:normalise} and~\ref{lem:conservative},
are satisfied.

\begin{lemma}[termination]\label{lem:normalise} The normalisation process
  terminates.
\end{lemma}

\begin{proof}
The number of nodes that can be processed by the three operations is finite, for every pre-proof tree. In addition, it
decrements after applying each operation.
\end{proof}

The induction function is extended to allow new 
bud-companion relations between nodes from different pre-proof trees.

\begin{definition}[rb-path, IH-node]
  \label{def:rb-path} An \emph{rb-path} is a path of the form
  $[R,\ldots, H, B]$ that leads the root $R$ to a bud $B$ in
  some pre-proof tree of a pre-proof tree-set such that $B$ is the
  only bud in the path.  We will call $H$ an inductive hypothesis node (for
  short, \emph{IH-node}).
\end{definition}

A path in a pre-proof tree-set ($\cal MD$, $\cal MR$) is a list
$[N^0, N^1, \ldots]$ of nodes in $\cal MD$ such that, for all
$i\geq 0$, $S(N^{i+1})$ is one of the premises of the rule applied on
$S(N^i)$ if $N^i$ is an internal node, or $S({\cal MR}(N^i))$ if
$N^i$ is a bud.

\begin{lemma}\label{lem:conservative} The  normalisation of any
  pre-proof $(\cal D, \cal R)$ of a sequent $S$ builds a pre-proof tree-set ($\cal MD$, $\cal MR$)
\begin{enumerate}
\item that has a
  pre-proof tree rooted by a node labelled by $S$, and
\item for which each of its rb-paths $[R,\ldots, B]$ has $B$ as the only
  node that is labelled by the premise of a $(Subst)$ rule. A node is
  a $(Subst)$-node if and only if it is an $(IH)$-node.
\end{enumerate}
\end{lemma}
\begin{proof}
  Claim 1) holds because the first operation duplicates only non-root
  nodes and the third operation expands bud nodes, so the root nodes
  do not change. If $S$ labels the root node of a pre-proof
  tree $t$ having a non-root companion $n$, $t$ will be processed by the second
  operation applied on $n$ but will still have its root labelled by $S$.

  Claim 2) holds by the construction of the normal forms.
\end{proof}

\begin{example}\label{ex:normalf} The second operation can be applied
  on the non-root companion from Example~\ref{ex:preproof}, denoted by
  $(*)$, to give the following normalised pre-proof tree-set:

\begin{center}
    \includegraphics[width=0.9\linewidth]{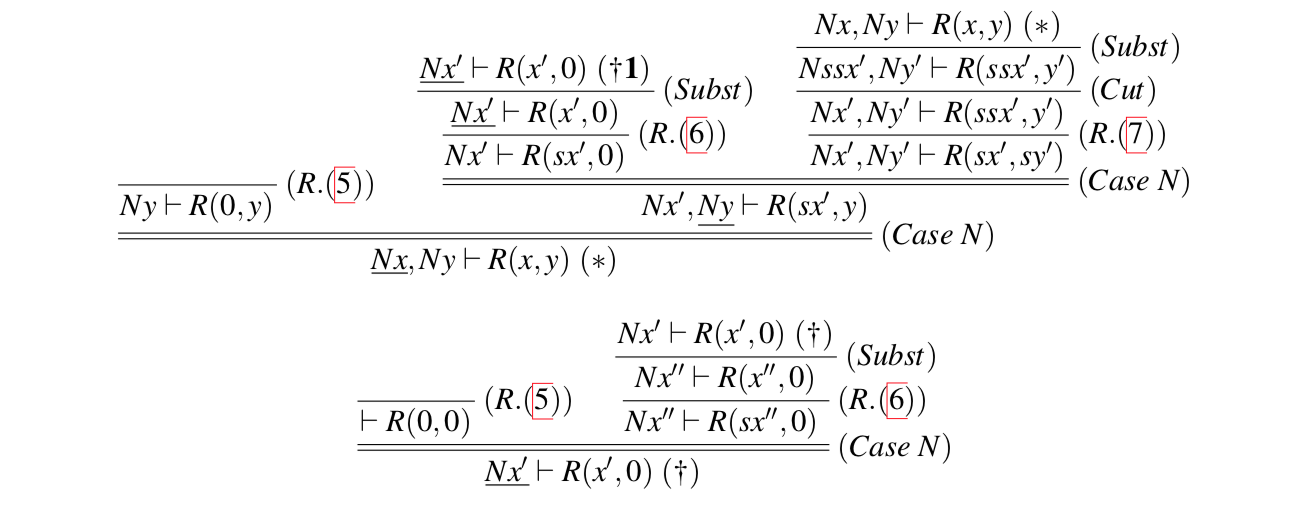}
  \end{center}


\end{example}

\subsection{Building the digraph of a pre-proof tree-set}
\label{sec:digraph}

Any pre-proof tree-set can also be represented as a \emph{digraph} of
sequents built from the nodes of its tree-set. The digraph associated
to a pre-proof tree-set ($\cal MD$, $\cal MR$) is crucial in our
setting to check
whether ($\cal MD$, $\cal MR$) is a proof tree-set.  Its edges are arrows built
as follows:
\begin{itemize}
\item a \emph{forward} arrow leads a node $N^1$ to a node $N^2$ if there is a
  rule that was applied on the sequent labelling $N^1$ and the sequent
  labelling $N^2$ is a premise of the rule;
\item a \emph{back-link} (or backward arrow)  starts from a bud  and
  ends to its companion. 
\end{itemize}

Some arrows will be annotated by substitutions. Each forward arrow,
starting from a $(Gen)$-node whose principal formula is $x=u$, is
annotated by the \emph{equality substitution} $\{x\mapsto u\}$.
The forward arrow starting from a node $N$ that is different from 
$(Gen)$- and $(Subst)$-nodes is annotated with the
\emph{identity substitution} for $S(N)$, which maps the free variables
from $S(N)$ to themselves. Finally, the forward arrows starting
from $(Subst)$-nodes and the back-links are not annotated. They help
to build infinite paths but do not play any role when defining the
soundness constraints.

By abuse of notation, a \emph{path} in a digraph is a (potentially
infinite) list of nodes built by following the arrows in the
digraph. An \emph{rb-path} is any path leading a root to some bud node
and does not have other bud nodes. \textbf{Unless otherwise stated, we
  will consider only rb-paths in the digraphs associated to \emph{normalised}
  pre-proof tree-sets.} 

\begin{remark}\label{rem:unique} According to Lemma~\ref{lem:conservative}, the
bud node $B$ of any such rb-path is the only node in the rb-path for
which $S(B)$ is the premise of a $(Subst)$ rule.
\end{remark}

\begin{definition}[cumulative substitution]\label{def:cumulative}
  An rb-path $[N^1,\ldots,N^n,B]$ ($n> 0$) can be annotated by
  the \emph{cumulative substitution}
  $\sigma_{id}^{all}\sigma_1\cdots\sigma_{n-1}$, where $\sigma_i$ is
  the substitution annotating the forward arrow leading $N_i$ to
  $N_{i+1}$, for each $i\in[1..n-1]$, and $\sigma_{id}^{all}$ is the
  \emph{overall} identity substitution
  $\cup_{N\in [N^1,\ldots,N^{n-1}]}\{x\mapsto x \mid x \in FV(S(N))\}$.
\end{definition}

A list of sequents $[ S_1, \ldots, S_n ]$ ($n>0$) is
 \emph{admissible} if either i) it is a singleton ($n=1$), or ii) for every $i\in[2..n]$, $S_i$ is the premise of some rule
 whose conclusion is $S_{i-1}$. By construction,
 the list of sequents labelling the nodes from every path from the digraph
 associated to a pre-proof tree-set is admissible. 

\begin{lemma}
  \label{lem:cumulative} Let $[N^1, \ldots, N^{n-1}, N^n, B]$ be an
  rb-path. We define its \emph{cumulative} list $l_c$ as
  $[ S(N^1)[\theta_{(1,n)}^c], \ldots, S(N^{n-1})[\theta_{(n-1,n)}^c],
  S(N^n),S(B)]$, where $\theta_{(i,n)}^c$ is the cumulative
    substitution for $[N^{i}, \ldots, N^{n-1}, N^n]$. Then, the
    following properties hold:

\begin{enumerate}
\item 
 $l_c$  is admissible,  and
\item the rule applied on each $S(N^{i})$ is also applicable on
  $S(N^{i})[\theta_{(i,n)}^c]$, $\forall i\in[1..n-1]$, if it is different
  from $(Gen)$. If the rule is $(Gen)$, the $(Gen)$-step can be replaced by a
  $(Wk)$-step, where the LK's $(Wk)$ rule is defined as
  \vspace{-.2cm}
  \begin{center}
    \includegraphics[width=0.9\linewidth]{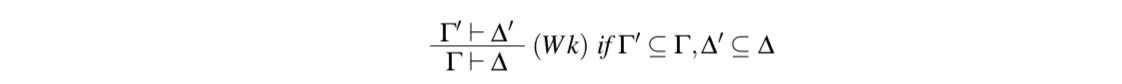}
  \end{center}
\end{enumerate}
\end{lemma}
  \vspace{-.4cm}
\begin{proof}
  We will perform induction on $n$. If $n=1$, then $N^n\equiv N^1$
  and $[S(N^1)]$ is a singleton, hence it is admissible.

  If $n>1$, let  $p$ denote the path 
  $[N^1, \ldots, N^{n-1},N^n,B]$. By induction hypothesis, we assume
  that
  $[ S(N^1)[\theta_{(1,n-1)}^{c}], \ldots, S(N^{n-2})[\theta_{(n-2,n-1)}^{c}],
  S(N^{n-1})]$
  is admissible, where $\theta_{(i,n-1)}^{c}$ ($i\in[1..n-2]$) is the cumulative
  substitution annotating $[N^{i}, \ldots, N^{n-1}]$ and the rules applied on
  $S(N^{i})[\theta_{(i,n-1)}^{c}]$ and $S(N^{i})$ are the same. We denote
  by $\theta_{(n-1,n-1)}^{c}$ the identity substitution for $S(N^{n-1})$.

  Let $\theta_{(i,n)}^c$ be the cumulative substitution annotating
  $[N^{i}, \ldots, N^{n-1}, N^n]$, for all $i\in[1..n-1]$. Let also
  $\theta$ be the substitution annotating the forward arrow leading
  $N^{n-1}$ to $N^n$, which can be
  either an identity substitution, or an equality substitution. In the
  first case, for every $i\in[1..n-1]$, $\theta_{(i,n)}^c$ is i)
  $\theta_{(i,n-1)}^{c}\cup \{x\mapsto x\mid x \in \overline{x}\}$ if
  the rule applied on $S(N^{n-1})$ is the LK's rule $(\forall R)$ or $(\exists L)$,
  defined below:
\begin{center}
    \includegraphics[width=0.9\linewidth]{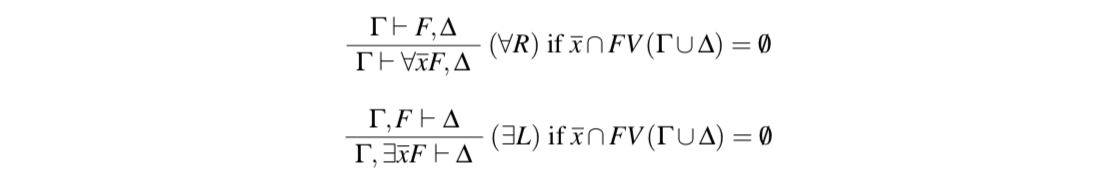}
  \end{center}
  
  and $\overline{x}$ is the vector of new free variables introduced by
  these rules, or ii)
  $\theta_{(i,n-1)}^{c}$, otherwise. Since
  $S(N^i)[\theta_{(i,n-1)}^{c}]\equiv S(N^i)[\theta_{(i,n)}^{c}]$ by
  induction hypothesis, we can apply the same
  rules on $S(N^i)[\theta_{(i,n)}^{c}]$ and $S(N^{i})$, hence the list
  $[ S(N^1)[\theta_{(1,n)}^c], \ldots, S(N^{n-1})[\theta_{(n-1,n)}^c],
  S(N^n)]$ is admissible. $[ S(N^1)[\theta_{(1,n)}^c], \ldots, S(N^{n-1})[\theta_{(n-1,n)}^c],
  S(N^n), S(B)]$ is also admissible since $S(B)$ is the premise of a
  $(Subst)$ rule whose conclusion is $S(N^n)$, by property 2) from Lemma~\ref{lem:conservative}.

  For the second case,  $\theta$ is an equality substitution.
  We have that $\theta_{(i,n)}^c$ equals $\theta_{(i,n-1)}^{c}\theta$,
  for all $i\in[1..n-1]$. Since the rule applied on a sequent can also
  be applied on every instance of it, we have that
  $[ S(N^1)[\theta_{(1,n)}^c], \ldots, S(N^{n-1})[\theta_{(n-1,n)}^c],
  S(N^n)]$ is admissible; the rule applied on $S(N^{i})$ can also be
  applied on $S(N^{i})[\theta_{(i,n)}^c]$, for all
  $i\in[1..n-1]$. Notice that the $(Gen)$ rule has $S(N^n)$ as premise
  when applied on $S(N^{n-1})[\theta_{(n-1,n)}^c\theta]$. Let us
  assume that $x=u$ is the principal formula of
  $S(N^{n-1})[\theta_{(n-1,n)}^c]$. Then, $\theta$ is
  $\{x\mapsto u\}$. On the one hand, $(Gen)$ cannot be applied on
  $S(N^{n-1})[\theta_{(n-1,n)}^c\theta]$, whose principal formula is
  $u=u$, when $u$ is a non-variable term. On the other hand, the
  generalised form of $(Gen)$ from CLKID$^\omega$, displayed in
  Figure~\ref{fig:eq}, would replace $u$ by
  $u$ and delete $u=u$. If $S(N^{n-1})[\theta_{(n-1,n)}^c\theta]$ is
  of the form $\Gamma,u=u\vdash \Delta$, the same result can be
  achieved with CLKID$_N^\omega$ by applying $(Wk)$ instead:
\begin{center}
    \includegraphics[width=0.9\linewidth]{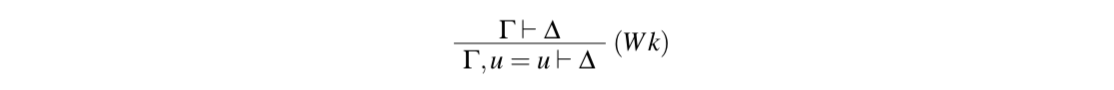}
  \end{center}

\sloppy{So, the list
$[ S(N^1)[\theta_{(1,n)}^c], \ldots, S(N^{n-1})[\theta_{(n-1,n)}^c],
S(N^n)]$ is admissible.

$[ S(N^1)[\theta_{(1,n)}^c], \ldots, S(N^{n-1})[\theta_{(n-1,n)}^c],
S(N^n), S(B)]$ is also  admissible}, as shown for the first case.
\end{proof}

A path has \emph{cycles} if some nodes are repeated in the path.  The
set of \emph{strongly connected components} (SCCs) of a digraph
$\cal P$ of some pre-proof tree-set ($\cal MD$, $\cal MR$) is a
partition of $\cal P$, where each SCC is a maximal sub-graph for which
any two different nodes are linked in each direction by following only
arrows from the sub-graph. Therefore, every non-singleton SCC has at
least one cycle. Additionally, if $\cal P$ is acyclic, each of its
nodes is a singleton SCC.

\begin{example}\label{ex:digraph} The digraph of the normalised pre-proof tree-set from Example~\ref{ex:normalf} is:

\begin{center}
    \includegraphics[width=.8\linewidth]{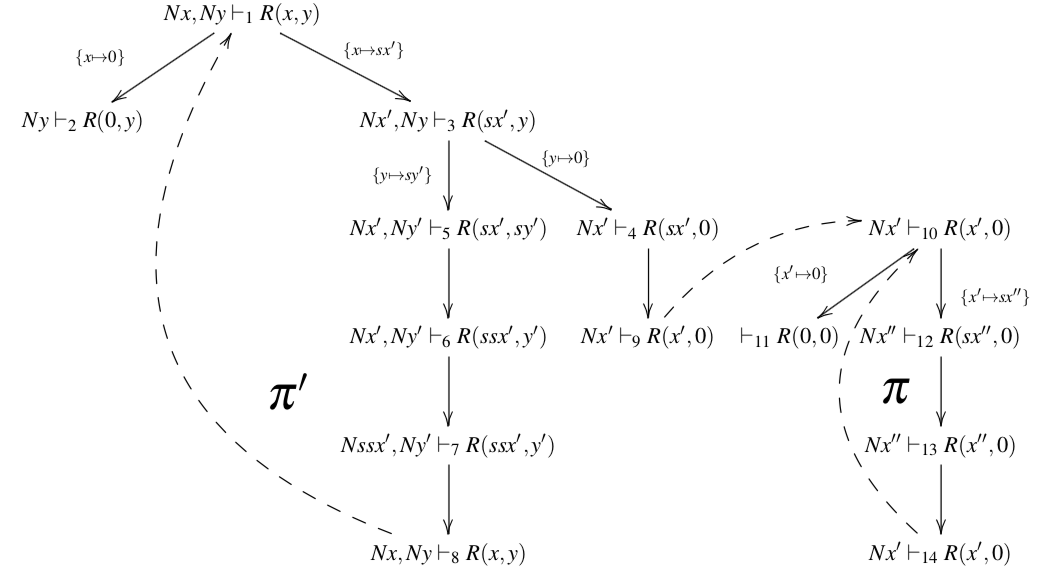}
  \end{center}



The sequent labelling a node is annotated by the number of the node in
the digraph. The digraph has two non-singleton SCCs: i) $\pi$:$\{N^{10},N^{12},N^{13}, N^{14}\}$,
and ii) $\pi':\{N^{1}, N^{3}, N^{5}, N^{6}, N^{7}, N^{8}\}$.
\end{example}

\subsection{Defining the ordering and derivability conditions}
\label{sec:validate}

The premises for defining the new soundness criterion are similar
to~\cite{Stratulat:2017ac}. 
Let $\pi$ be a SCC from $\cal P$ and $<_a$ an ordering \emph{stable
  under substitutions} defined over the set $\cal S$ of instances of
the IAAs from the sequents labelling nodes inside $\pi$, i.e., if
$l <_a l'$ then $l[\sigma] <_a l'[\sigma]$, for all $l,l'\in {\cal S}$
and substitution $\sigma$. 
Given a path $p$ in $\pi$, we say that an IAA $\tau_j$
\emph{derives from} an IAA $\tau_i$ using the trace
$(\tau_k)_{(k\geq 0)}$ along $p$ if $i <j$. Also, given two arbitrary
substitutions $\gamma$ and $\delta$, we say that $\tau_j[\gamma]$
derives from $\tau_i[\delta]$ using $(\tau_k)_{(k\geq 0)}$ along
$p$. $<_{\pi}$ is  the \emph{multiset
  extension}~\cite{Baader:1998ve} of $<_a$.

The ordering constraints from a multiset extension relation comparing
two sequent instances can be combined with derivability constraints on
IAAs to give the \emph{$<_{\pi}$-derivability} relation, referred to as
\emph{ordering-derivability} when the ordering is not known. For this,
we assume that every sequent $S$ has associated a measure value
(weight), denoted by $A_{S}$ and represented by a multiset
of IAAs of $S$. 

\begin{definition}[$<_{\pi}$-derivability]
\label{def:derivabillity}
Let $N^i$ and $N^j$ be two nodes occurring in some path $p$ from
$\pi$, and $\theta$, $\delta$ be two substitutions. We define
$A_{S(N^i) [\theta]}'$ (resp., $A_{S(N^j)[\delta]}'$) as the multiset,
resulting from $A_{S(N^i) [\theta]}$ (resp., $A_{S(N^j)[\delta]}$)
after the pairwise deletion of all common IAAs from
$A_{S(N^i) [\theta]}$ and $A_{S(N^j) [\delta]}$. In addition, we
assume that for each
$l\in A_{S(N^j)[\delta]}\mybs A_{S(N^j)[\delta]}'$, there is
$l'\in A_{S(N^i) [\theta]}$\textbackslash$ A_{S(N^i) [\theta]}'$
satisfying i) $l\equiv l'$, and ii)
$l$ is the unique literal from $A_{S(N^j)[\delta]}$ that derives from $l'$ using some trace following $p$. 

Then,
$S(N^j)[\delta]$ is \emph{$<_{\pi}$-derivable} from $S(N^i)[\theta]$
along $p$ if 
for each $l\in A_{S(N^j)[\delta]}'$ there exists
$l'\in A_{S(N^i) [\theta]}'$ such that $l'>_a l$ and $l$ derives from
$l'$ using some trace following $p$.
\end{definition}



By the definition of $<_{\pi}$ as a multiset extension of $<_a$, the
following results 
can be proved when considering some path in $\pi$.

\begin{lemma}\label{lem:derivable}
  If $S$ is $<_{\pi}$-derivable from $S'$ then $A_S<_{\pi}A_{S'}$.
\end{lemma}

\begin{proof}
By the definition of the ordering constraint in the
$<_{\pi}$-derivability relation. 
\end{proof}

\begin{lemma}\label{lem:derivable-stability}
The `$<_{\pi}$-derivability' relation is stable under substitutions
and transitive.
\end{lemma}

\begin{proof}
  Let $S$ and $S'$ be two sequents such
  that $S$ is $<_{\pi}$-derivable from $S'$ along some path $p$ in $\pi$. By
  Lemma~\ref{lem:derivable}, $A_{S'}>_{\pi}A_S$. Since $<_{\pi}$ is stable
  under substitutions, we have that $A_{S'[\sigma]}>_{\pi}A_{S[\sigma]}$, for every
  substitution $\sigma$. According to
  Definition~\ref{def:derivabillity}, the derivability relations
  between their IAAs do not change by instantiation
  operations. Therefore, $S[\sigma]$ is $<_{\pi}$-derivable from
  $S'[\sigma]$ along $p$. We conclude that
  the `$<_{\pi}$-derivability' relation is stable under substitutions.\\

To prove the transitivity property, let us assume three sequents
$S_1$, $S_2$ and $S_3$ labelling nodes in a path $p$ built by the
concatenation of two paths $p_1$ and $p_2$ such
that $S_3$ is
$<_{\pi}$-derivable from  $S_2$ along $p_2$ and $S_2$ is
$<_{\pi}$-derivable from  $S_1$ along $p_1$. We will try to prove that $S_3$ is
$<_{\pi}$-derivable from  $S_1$ along $p$.

Since $S_3$ is
$<_{\pi}$-derivable from  $S_2$ along $p_2$,  by
Definition~\ref{def:derivabillity} we have that
\begin{enumerate}[label=(\roman*1)]
\item for each
$l_3\in A_{S_3}'$ there exists $l_2\in A_{S_2}'$ such that
$l_2>_a l_3$ and $l_3$ derives from $l_2$ using some trace 
following $p_2$, and 
\item for each
$l_3\in A_{S_3}$\textbackslash$ A_{S_3}'$, there is some
$l_2\in A_{S_2}$\textbackslash$ A_{S_2}'$ such
that $l_3\equiv l_2$ and
$l_3$ is the unique IAA that derives from $l_2$ using some trace following $p_2$,
\end{enumerate}
where $A_{S_3}'$ (resp., $A_{S_2}'$) is the multiset
resulting from $A_{S_3}$ (resp., $A_{S_2}$)
after the pairwise deletion of all common IAAs
from $A_{S_3}$ and $A_{S_2}$. Also, since $S_2$ is
$<_{\pi}$-derivable from  $S_1$ along $p_1$, we have that
\begin{enumerate}[label=(\roman*2)]
\item
for each
$l_2\in A_{S_2}''$, there exists $l_1\in A_{S_1}'$ such that
$l_1>_a l_2$ and $l_2$ derives from $l_1$ using some trace 
following $p_1$, and
\item for each
$l_2\in A_{S_2}$\textbackslash$ A_{S_2}''$, there is some
$l_1\in A_{S_1}$\textbackslash$ A_{S_1}'$ such
that $l_2\equiv l_1$ and
$l_2$ is the unique IAA that derives from $l_1$ using some trace following $p_1$,
\end{enumerate}
where $A_{S_2}''$ (resp., $A_{S_1}'$) is the multiset
resulting from $A_{S_2}$ (resp., $A_{S_1}$)
after the pairwise deletion of all common IAAs
from $A_{S_2}$ and $A_{S_1}$. We have to check that 
for each
$l_3\in A_{S_3}''$, there exists $l_1\in A_{S_1}''$ such that
$l_1>_a l_3$ and $l_3$ derives from $l_1$ using some trace
following $p$, 
where $A_{S_3}''$ (resp., $A_{S_1}''$) is the multiset
resulting from $A_{S_3}$ (resp., $A_{S_1}$)
after the pairwise deletion of all common IAAs
from $A_{S_3}$ and $A_{S_1}$. Moreover, for each
$l_3\in A_{S_3}$\textbackslash$ A_{S_3}''$, there is some
$l_1\in A_{S_1}\mybs A_{S_1}''$ such
that $l_3\equiv l_1$ and
$l_3$ is the unique IAA that  derives from $l_1$ using some trace  following $p$. We consider the following cases:

\begin{enumerate}
\item
  If $l_3 \in A_{S_3}'$ there exists $l_2\in A_{S_2}'$ such that
  $l_2 >_a l_3$ and $l_3$ derives from $l_2$ using some trace $t_2$
  following $p_2$. 
  \begin{enumerate}
  \item If
    $l_2 \in A_{S_2}''$ there exists $l_1\in A_{S_1}'$ such that
    $l_1 >_a l_2$ and $l_2$ derives from $l_1$ by using some trace
    $t_1$ following $p_1$. Then $l_1 >_a l_3$ by
      the transitivity of $<_a$, so $l_1 \in A_{S_1}''$, $l_3 \in A_{S_3}''$ and $l_3$ derives from $l_1$ using the
      concatenation of $t_1$ and $t_2$ following $p$.
    \item If $l_2\in A_{S_2}$\textbackslash$ A_{S_2}''$, there is
      $l_1\in A_{S_1}$\textbackslash$ A_{S_1}'$ such that
      $l_2\equiv l_1$ and $l_2$ is the unique IAA that derives from $l_1$ by using some trace
      $t_1$ following $p_1$. Since $l_1 (\equiv l_2) >_a l_3$, we have
      that $l_1 \in A_{S_1}''$, $l_3 \in A_{S_3}''$ and $l_3$ derives from $l_1$ using the
      concatenation of $t_1$ and $t_2$ following $p$.
  \end{enumerate}

\item If $l_3 \in A_{S_3}$\textbackslash$ A_{S_3}'$ there exists $l_2\in A_{S_2}'$ such that
  $l_3 \equiv l_2$ and $l_3$ is the unique IAA that derives from $l_2$ using some trace $t_2$
  following $p_2$. 
  \begin{enumerate}
  \item If
    $l_2 \in A_{S_2}''$ there exists $l_1\in A_{S_1}'$ such that
    $l_1 >_a l_2$ and $l_2$ derives from $l_1$ by using some trace
    $t_1$ following $p_1$. Then, $l_1 >_a (l_2 \equiv) l_3$, so $l_1
    \in A_{S_1}''$, $l_3 \in A_{S_3}''$ and $l_3$ derives from $l_1$ using the
      concatenation of $t_1$ and $t_2$ following $p$.
    \item If $l_2 \in A_{S_2}$\textbackslash$ A_{S_2}''$ there exists
      $l_1\in A_{S_1}'$ such that $l_1 \equiv l_2$ and $l_2$ is the
      unique IAA that derives from $l_1$ by using some trace $t_1$
      following $p_1$. This means that
      $l_3\in A_{S_3}$\textbackslash$ A_{S_3}''$,
      $l_1\in A_{S_1}$\textbackslash$ A_{S_1}''$ with $l_1\equiv (l_2
      \equiv) l_3$
      and $l_3$ derives from $l_1$ using the concatenation of $t_1$
      and $t_2$ following $p$. In addition, $l_3$ is the unique IAA in
      $A_{S_3}$ that derives from $l_1$.
  \end{enumerate}
  
\end{enumerate}

\end{proof}

The soundness criterion consists in
checking if the sequents labelling $(IH)$-nodes from every
non-singleton SCC, referred to as \emph{induction hypotheses},
satisfy some constraints.

\begin{definition}[induction hypothesis (IH), IH discharged by a
 SCC] 
\label{def:IHcycle}
Let $\pi$ be a non-singleton SCC and  $[R,\ldots, H, B]$ an rb-path $p$ in
$\pi$. We say that the \emph{induction
hypothesis} (IH) $S(H)$ is \emph{discharged} by $\pi$ if $S(H)$ is 
$<_{\pi}$-derivable from $S(R)[\theta^c]$ along
$p$, where $\theta^c$ is the cumulative substitution
annotating $p$.
\end{definition}

\begin{theorem}[soundness]
\label{thm:soundness} 
The sequents, labelling the roots from every
normalised pre-proof tree-set whose non-singleton SCCs  discharge
their IHs,  are true.
\end{theorem}

\begin{proof}
Let $M$ be a standard model for $(\Sigma,\Phi)$ and assume a
  normalised pre-proof tree-set. Let also $\cal P$ denote its digraph whose non-singleton SCCs discharge
  their IHs. By contradiction, we assume that there exists a root node $N$ such
  that $S(N)$ is false.  
  We define a partial (well-founded)
  ordering $<_{\cal R}$ over the (finite number of) root nodes from
  $\cal P$ such that, for every two distinct root nodes $N^1$ and $N^2$, we
  have $N^1 <_{\cal R} N^2$ if i) $N^1$ and $N^2$ are not in the same
  SCC, and ii) $N^1$ can be joined from $N^2$
  in $\cal P$.

  By induction on $<_{\cal R}$, we consider the base case when $N$ is a
  $<_{\cal R}$-minimal node. (The step case, when $N$ is not a
  $<_{\cal R}$-minimal node, will not be detailed since it can be
  treated similarly by assuming that all $<_{\cal R}$-smaller root
  nodes are labelled by true sequents.) If $N$ is included in a
  one-node SCC, $N$ is also a leaf node. The only 0-premise rules are
  the LK's
  $(Ax)$ rule 
 as well as $(R.)$
when unfolding with unconditional axioms. In both cases, $S(N)$ is true
which leads to a contradiction.

Let us now assume that $N$ is a $<_{\cal R}$-minimal node from some
non-singleton SCC $\pi$. We will analyse all possible scenarios and show that each of
them leads to a contradiction.
The tree $t$ from $\cal P$ and rooted by $N$ should have buds labelled
by false sequents, otherwise $S(N)$ would be true. Let $B$ be such a
bud such that $N^h$ is its companion and $[N,\ldots,H,B]$ is an
rb-path in $\pi$. $N^h$ should be a root node from $\pi$ because $N$
is $<_{\cal R}$-minimal; it is labelled by the false
sequent $S(B)$. 
Since the 
CLKID$^\omega_N$ rules are sound, by  Lemma~\ref{lem:cumulative}, we
conclude that the
\emph{cumulative instance} $S(N)[\theta_c]$ is false, where  $\theta_c$ is the
cumulative substitution for $[N,\ldots,H,B]$. $\pi$ discharges
its IHs, so we have that
$S(B)[\delta_h](\equiv S(H))$ is $<_{\pi}$-derivable
from $S(N)[\theta_c]$, where $\delta_h$ is the substitution used by
the $(Subst)$-step whose conclusion is $S(H)$. 
By Lemma~\ref{lem:derivable}, we
have that $A_{S(N^h)[\delta_h]} <_{\pi} A_{S(N)[\theta_c]}$. 

We perform a similar reasoning on $N^h$ as for $N$. There is an rb-path
$[N^h,\ldots, H', {N^f}']$ such that the companion of ${N^f}'$ (in $\pi$) is ${N^h}'$
and $S(N^h)[\delta_h]$ shares false instances with
$S(N^h)[\theta^c_1]$, where $\theta^c_1$ is the cumulative
substitution annotating 
$[N^h,\ldots, H', {N^f}']$. By contradiction, we assume that no false
instance of $S(N^h)[\delta_h]$ is shared. Then, one can
build a finite bud-free pre-proof tree  of $S(N^h)[\delta_h]$, by using
only sound rules. Hence,
$S(N^h)[\delta_h]$ is true, so contradiction.
Therefore, there are
two substitutions $\epsilon$ and $\tau$ such that
$S(N^h)[\delta_h\epsilon]\equiv S(N^h)[\theta^c_1\tau]$ and
$S(N^h)[\theta^c_1\tau]$ is false. Let $S({N^h}')[\delta_h'] (\equiv S(H'))$ be the
instance of $S({N^h}')$ used as IH. Since it is discharged
by $\pi$, we have that
$A_{S(N^h)[\theta^c_1]}>_{\pi}A_{S({N^h}')[\delta_h']}$. From $A_{S(N)[\theta^c]}
>_{\pi} A_{S(N^h)[\delta_h]}$ and the previous ordering
constraint, we get
$A_{S(N)[\theta^c\epsilon]} >_{\pi} A_{S(N^h)[\delta_h\epsilon]}$ and
$A_{S(N^h)[\theta^c_1\tau]}>_{\pi}A_{S({N^h}')[\delta_h'\tau]}$, by the `stability
under substitutions' property of $<_{\pi}$. Hence,
\[A_{S(N)[\theta^c\epsilon]} >_{\pi} A_{(S(N^h)[\delta_h\epsilon]}
  ~\equiv)~A_{S(N^h)[\theta^c_1\tau]}>_{\pi}A_{S({N^h}')[\delta_h'\tau]}\]
For similar reasons as given for $S(N^h)[\delta_h]$, we can show that
$S({N^h}')[\delta_h'\tau]$ is false, hence it can be treated similarly
as $S(N^h)[\delta_h]$. And so on, the process
can be repeated to produce an  infinite strictly
$<_{\pi}$-decreasing sequence $s$ of measure values associated to instances of sequents labelling
root nodes from $\pi$, of the form
\[A_{S(N)[\theta^c\epsilon\cdots]} >_{\pi}
  A_{S(N^h)[\theta^c_1\tau\cdots]}>_{\pi}A_{S({N^h}')[\cdots]}>_{\pi}\ldots\hspace{1cm}\]

We can associate to $s$ the infinite admissible list $l_s$ of its sequents $[S(N)[\theta^c\epsilon\cdots], S(N^h)[\theta^c_1\tau\cdots], S({N^h}')[\cdots],\ldots]$
and define the path $p$ \emph{underlying} $l_s$ as the concatenation of the rb-paths from $\pi$ that  built $s$, i.e., $[N,\ldots,B,N^h,\ldots,
N^{f'}, \ldots]$. By the construction of $s$, every successive $(Subst)$-, bud and
root nodes in $p$ are labelled by the same sequent instance in $l_s$, so the
$(Subst)$-steps are stuttering in $l_s$. By
Lemma~\ref{def:cumulative}, all $(Gen)$- can be replaced by $(Wk)$-steps. $p$ is
  of the form $[N_{\infty}\ldots,N_1,\ldots, N_0]$ where $N_0$, $N_1$, \ldots,
  $N_{\infty}$ are an infinite number of \emph{all} the occurrences of
  $N$ in $p$.

  We will show that there is a trace
  following $p$ that has an infinite number of progress points. As
  explained in~\cite{Brotherston:2011fk}, it means that there is
  an infinite strictly decreasing sequence of ordinals, hence contradiction. Since
  $p$ is the concatenation of rb-paths in $\pi$ and $\pi$ discharges
  its IHs, for each such rb-path the bud sequent is
  $<_{\pi}$-derivable from the cumulative instance, along the rp-path, of the root
  sequent. By Lemma~\ref{lem:derivable-stability}, there is an
  instance $S(N_{\infty})[\theta_{\infty}]$ such that $S(N_0)$ is
  $<_{\pi}$-derivable from it along $p$, where $\theta_{\infty}$ is
  the composition of all cumulative substitutions of the rb-paths from
  $l$. For any two consecutive nodes $N_i$ and $N_{i-1}$
  ($i\in[1..\infty]$), we have that $S(N_{i-1})[\theta_{i-1}]$ is
  $<_{\pi}$-derivable from $S(N_{i})[\theta_{i}]$, where $\theta_i$
  (resp., $\theta_{i-1}$) are the compositions of all cumulative
  substitutions of the rb-paths along $[N_i,\ldots, N_0]$ (resp.,
  $[N_{i-1},\ldots, N_0]$).

  Let us denote by $S$ (resp, $S'$) the sequent $S(N_{i})[\theta_{i}]$
  (resp., $S(N_{i-1})[\theta_{i-1}]$), for some $i\in[1..\infty]$.  By
  Definition~\ref{def:derivabillity} and the transitivity of the
  $<_{\pi}$-derivability relation, for each IAA $l$ from $A_S$ there
  is an IAA $l'$ from $A_{S'}$ such that $l$ derives from
  $l'$. Therefore, there are $n$ traces along the path $p'$
  $[N_{\infty},\ldots,N_i]$, where $n$ is the number of IAAs from $S$.

  We will show that the traces along $p'$ have an infinite number of
  progress points. By contradiction, we assume that this number is
  finite. Therefore, there is a subpath $p''$ of $p$ whose traces have
  no progress points and there exists $j\in[1..\infty]$ such that
  $N_j$ and $N_{j-1}$ belong to $p''$. Let us denote by $S_j$ (resp,
  $S_{j-1}$) the sequent $S(N_{j})[\theta_{j}]$ (resp.,
  $S(N_{j-1})[\theta_{j-1}]$). Since $S_{j-1}$ is $<_{\pi}$-derivable
  from $S_j$, we have that $A_{S_{j-1}}<_{\pi}A_{S_{j}}$. By the definition of
  $<_{\pi}$ as a multiset extension of the ordering $<_a$ over the
  instances of IAAs from the root sequents in $\pi$, there should be
  an IAA $l\in A_{S_{j-1}}$ for which there is another IAA
  $l\in A_{S_j}$ such that $l<_a l'$ and $l$ derives from $l'$, i.e.,
  $l$ and $l'$ are from an infinite trace $t$ following (a subpath of)
  $p''$ which has no progress points. According to the definition of a
  trace (see Definition~\ref{def:trace}) and the way $l_s$ was built,
  $l <_a l'$ is possible only if the subtrace of $t$ from $l'$ to $l$
  has at least one progress point, so contradiction. Otherwise,
  $l\equiv l'$ since i) $l_s$
  is admissible, ii) the $(Subst)$-steps are stuttering, iii) the
  $(Gen)$-steps can be replaced by $(Wk)$-steps, and iv) the
  instantiation steps that built $s$ preserve the
  equality relations.
\end{proof}

\begin{example}
  For the sequents labelling the nodes from the digraph given in
  Example~\ref{ex:normalf}, we define the measure values $A_{Nt\vdash R(t,0)}=\{Nt\}$,
  $\forall t$, and $A_{Nt_1,Nt_2\vdash R(t_1,t_2)}=\{Nt_2\}$,
  $\forall t_1, t_2$. The IH $S(N^{13})$ is $<_{\pi}$-derivable from
  $S(N^{10})[{\{x'\mapsto sx''\}}]$, hence discharged by the SCC $\pi$ using the trace
  $[Nx',Nx'',Nx'']$, if
  $\{Nx''\} <_{\pi} \{Nsx''\}$. Also, the IH $S(N^{7})$ is
  $<_{\pi'}$-derivable from
  $S(N^{1})[{\{x\mapsto sx'; y\mapsto sy'\}}]$ in the SCC $\pi'$ using the
  trace $[Ny,Ny,Ny',Ny',Ny', Ny']$ if
    $\{Ny'\} <_{\pi'} \{Nsy'\}$. The ordering constraints hold if
    $<_{\pi}$ and $<_{\pi'}$ are the  multiset extensions of a
    recursive path ordering~\cite{Baader:1998ve} $<_{rpo}$ for which $z
    <_{rpo}sz$, for every variable $z$.

    By Theorem~\ref{thm:soundness},  the root sequents
    in the pre-proof tree-set, $S(N^1)$ and $S(N^{10})$, are true.
\end{example}

\noindent\textbf{Comparison with the soundness checking criterion from~\cite{Stratulat:2017ac}.}
In~\cite{Stratulat:2017ac}, the ordering-derivability constraints
issued when analysing if a pre-proof tree-set is a proof are
defined at the level of the minimal cycles of its digraph, referred to as $n$-cycles. A
\emph{$n$-cycle} is defined as a finite circular list
$[N_1^1,\ldots,N_1^{p_1}],\ldots,$ $[N_n^1,\ldots,N_n^{p_n}]$ of $n$
($>0$) paths leading root nodes to buds such that
$N^1_{next(i)}={\cal MR}(N^{p_i}_i)$, for any $i\in[1..n]$, where
$next(i) = 1 + (i\mbox{ mod }n)$.

Let $\pi$ be a non-singleton SCC and $C$ an \emph{$n$-cycle}
$[N_1^1,\ldots,N_1^{p_1}],\ldots,$ $[N_n^1,\ldots,N_n^{p_n}]$ from
$\pi$. The induction hypotheses are defined at the $n$-cycle
level. For all $i\in[1..n]$, let $\theta_i^c$ be the cumulative
substitution annotating $[N_i^1,\ldots,N_i^{f}]$, where the
\emph{IH-node} $N^{f}_{i}$ is either i) $N_i^{p_i}$ if $(Subst)$ is
not applied along $[N_i^1,\ldots,N_i^{p_i}]$, or ii) $N_i^{p_i-1}$,
otherwise. The sequents labelling the IH-nodes correspond exactly to
the induction hypotheses used in the paper. We say that the IHs
$S(N_j^f)$ ($j\in[1..n]$) are \emph{discharged} by $C$ if,
$\forall i\in[1..n]$, $S(N^{f}_{i})$ is $<_{\pi}$-derivable from
$S(N_i^1)[\theta_i^c]$ along
$[N_i^1,\ldots,N_i^{p_i}]$. In~\cite{Stratulat:2017ac}, a \emph{proof}
is every pre-proof tree-set whose digraph has only $n$-cycles that
discharge their IHs and it has been shown that its root sequents  are true.

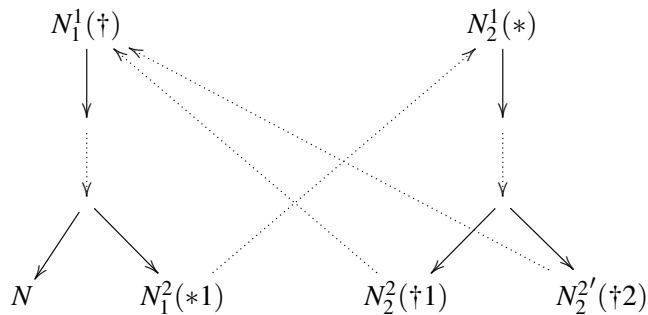
\begin{figure}[!h]
\centering
\begin{displaymath}
    \xymatrix @C=0.1pc {
      & N_1^1(\dag) \ar@{->}[d] & & ~~~~~~~~~~~~~~& & N_2^1 (*)\ar@{->}[d]  &\\
     & \ar@{.>}[d] & & & &  \ar@{.>}[d] &\\
       & \ar@{->}[dr]\ar@{->}[dl] & & & & \ar@{->}[dr] \ar@{->}[dl]&\\
    N &  & N_1^2 (*1)\ar@{.>}[uuurrr]  & &  N_2^2 \ar@{.>}[uuulll] (\dag 1)
    &  &{N_2^2}' (\dag 2)\ar@{.>}[uuulllll]
 }
\end{displaymath}
\caption{\label{fig:cycle} Two $2$-cycles sharing the same path.}
\end{figure}

Since several $n$-cycles may share the \emph{same} root-bud path, some
ordering-derivability constraints may be duplicated when checking that
a pre-proof is a proof. For example, the path $[N_1^1,\ldots, N_1^2]$
is shared between the two $2$-cycles
$[N_1^1,\ldots, N_1^2][N_2^1,\ldots, N_2^2]$ and
$[N_1^1,\ldots, N_1^2][N_2^1,\ldots, {N_2^2}']$ from the digraph given
in Figure~\ref{fig:cycle}.  Even if the number of $n$-cycles from a
digraph can be large, as explained in the introduction, the number of
\emph{distinct} ordering-derivability constraints is always smaller or
equal than the number of buds from the non-singleton SCCs. With the
approach from~\cite{Stratulat:2017ac}, the duplicates of the
constraints do not need to be again processed if the already processed
constraints are recorded.  It has been shown that the time complexity
of the soundness checking procedure is polynomial if the number of the
ordering-derivability constraints is that of the buds from the
non-singleton SCCs. With our new approach, the number of operations
for normalising a CLKID$^{\omega}_N$ pre-proof of $n$ nodes is given
by the sum of non-root companions, non-terminal $(Subst)$-nodes and
nodes labelled by some sequent that is the premise of a rule $r$
different from $(Subst)$. So, it is smaller than 3$n$. Let $c$ be the
maximal cost of an operation, including the node
duplication and the creation of a $(Subst)$-node or bud-companion
relation. Their total cost is smaller than 4$nc$ (the second operation
duplicates it twice). The costs for annotating substitutions and for
evaluating an ordering-derivability constraint are given
in~\cite{Stratulat:2017ac}.

\section{Conclusions and future work}

We have defined a more efficient soundness
criterion for a class of CLKID$^{\omega}$ pre-proofs considered
in~\cite{Stratulat:2017ac}, by building a set of non-redundant
ordering-derivability constraints. We have shown that these
constraints can also be extracted from those that define the soundness
criterion from~\cite{Stratulat:2017ac}, by deleting the duplicated
values. The new normal forms and their digraphs allow to uniformly
represent rb-paths and can be built in linear time. We conclude that
the two soundness checking criteria have the same polynomial-time
complexity if the time complexity for comparing two IAAs is at most
polynomial.

In the future, we plan to adapt our approach to make more effective
other soundness criteria based on minimal cycles, e.g., those
involving cyclic formula-based Noetherian induction
reasoning~\cite{Stratulat:2012uq,Stratulat:2017aa}, and other systems
where the soundness can be checked by the global trace condition, as
CLJID$^{\omega}$~\cite{Berardi:2018aa}.

\section*{Acknowledgements}

The author thanks the anonymous reviewers for their comments that
helped to improve the quality of the paper.

 \end{document}